\documentclass[a4paper,11pt]{article}

\usepackage[margin=2.5cm]{geometry}
\usepackage{authblk} 
\usepackage[labelfont=bf,labelsep=period]{caption}
\usepackage[nodayofweek]{datetime}
\usepackage{enumitem}
\usepackage{thmtools}
\usepackage{thm-restate}
\usepackage{float}
\usepackage{graphicx}
\usepackage{hyperref}
\usepackage[utf8]{inputenc}
\usepackage{numbertabbing}
\usepackage{verbatim}
\usepackage{dsfont} 
\usepackage{url}
\usepackage{xcolor}
\usepackage{xspace}
\usepackage{tikz}
\usepackage{graphicx}
\usetikzlibrary{trees, positioning}
\usepackage{colortbl}
\usetikzlibrary{patterns}
\usepackage{amsmath} %%% better math %%%
\allowdisplaybreaks[2]          % but allow some broken displays
\usepackage{amssymb} %%% symbols %%%
\usepackage{amsthm} %%% Theorem environments with `amsthm' %%%
\usepackage[table]{xcolor}
\usepackage{array}

\newdateformat{simple}{\THEDAY\ \monthname[\THEMONTH]\ \THEYEAR}
\simple

\floatstyle{ruled}
\newfloat{algo}{htbp}{algo}
\floatname{algo}{Algorithm}

\theoremstyle{definition}
\newtheorem{definition}{Definition}[section]
\theoremstyle{lemma}
\newtheorem{lemma}{Lemma}[section]
\theoremstyle{theorem}

\theoremstyle{corollary}
\newtheorem{corollary}{Corollary}[section]
\theoremstyle{example}
\newtheorem{example}{Example}[section]

\def \ifempty#1{\def\temp{#1} \ifx\temp\empty }

\newcommand{\var}[1]{\textit{#1}}
\newcommand{\op}[1]{\textsl{#1}}

\newcommand{\CF}{\ensuremath{\mathcal{F}}\xspace}

\newcommand{\CP}{\ensuremath{\mathcal{P}}\xspace}
\newcommand{\CQ}{\ensuremath{\mathcal{Q}}\xspace}

\newcommand{\CS}{\ensuremath{\mathcal{S}}\xspace}

\newcommand{\BF}{\ensuremath{\mathbb{F}}\xspace}

\newcommand{\etal}{\emph{et al.}}

\definecolor{lightorange}{RGB}{255,220,180}
\definecolor{lightred}{RGB}{255, 180, 180}

\title{\bf Symmetry all the way down}

\author[1]{Ignacio Amores-Sesar}
\author[2]{Christian Cachin}
\author[3]{Simon Holmgaard Kamp}
\author[2]{Juan Villacis}
\affil[1]{Aarhus University, Denmark}
\affil[2]{University of Bern, Switzerland}
\affil[2]{Ruhr University Bochum, Germany}

\date{} 

\begin{document}

\maketitle
\begin{abstract}

Asymmetric trust generalizes classical symmetric quorum systems by allowing each process to specify its own failure assumptions. While this flexibility enables tolerance of strictly more failure scenarios, it is not known if, in these cases, it is actually possible to solve distributed tasks, and if so, which. 

We answer this question using the depth hierarchy for asymmetric trust (Amores-Sesar et al., OPODIS~'25), which characterizes how much a process must rely on others to solve a task. We prove that asymmetric trust does not increase the solvability of tasks requiring depth two or more, such as reliable broadcast or consensus. Specifically, for any Byzantine asymmetric quorum system, every failure scenario that permits solving a task requiring depth at least two can also be tolerated by a suitably constructed Byzantine symmetric quorum system. We show this via a compiler that transforms asymmetric quorum systems into symmetric ones. 

The additional failure patterns tolerated exclusively by asymmetric trust correspond  to scenarios in which only simpler tasks requiring depth one or less (such as consistent broadcast) can be solved. We further prove that this result is tight in the depth hierarchy, meaning that there exist no compilers that produce symmetric quorum systems that are valid also in failure scenarios where correct processes have depths one or less. 

Our results clarify the precise power of asymmetric trust. While it strictly enlarges the set of tolerable failure patterns, it does not provide additional strength for solving tasks requiring depth two or higher.

% mention what can be solved with a symmetric qs that satisfies consistency and availability

\end{abstract}

%\tableofcontents

\section{Introduction}

Solving problems in distributed systems requires making explicit assumptions about the types and patterns of failures that may occur. Most classical systems rely on global adversarial thresholds, which assume that at most $f$ out of $n$ processes may be faulty at any given time. A more general abstraction is provided by symmetric quorum systems~\cite{DBLP:journals/dc/MalkhiR98}. Instead of bounding the number of faulty processes, quorum systems characterize the sets of processes that may concurrently fail, often based on shared properties such as geographic location, operating system, or administrative domain. Under these models, all participants adopt the same trust assumptions. Accordingly, these are referred to as \emph{symmetric trust assumptions}, and protocols built upon them as \emph{symmetric protocols}.

However, the symmetric model does not capture the heterogeneous and subjective trust relationships that occur in the real world. To address this limitation, \emph{asymmetric trust} was introduced~\cite{DBLP:conf/asiacrypt/DamgardDFN07}. In this setting, each process operates under its own trust assumptions, formalized through an asymmetric quorum system. Processes may define their quorums independently, reflecting individual views of which failures are expected. Protocols designed for this environment are known as asymmetric protocols. 

Another benefit provided by asymmetric trust is that, in fully permissionless environments~\cite{DBLP:journals/corr/abs-2304-14701}, it provides an alternative to traditional consensus protocols~\cite{nakamoto2008bitcoin, DBLP:conf/podc/PassS17} that mitigate Sybil attacks by requiring external resources, such as computational power or capital. A prominent example is the Stellar Consensus Protocol (SCP)~\cite{mazieres2015stellar}, which relies exclusively on local, subjective trust decisions, from which global asymmetric quorums emerge organically. By decoupling security from physical or economic resources, asymmetric trust enables consensus in permissionless networks through a completely different approach, based solely on trust relationships, which more closely resembles traditional permissioned BFT algorithms. This gives rise to more efficient and fast protocols.  

Asymmetric trust has been extensively studied~\cite{DBLP:conf/asiacrypt/DamgardDFN07, mazieres2015stellar, DBLP:conf/wdag/LosaGM19, DBLP:conf/opodis/SheffWRM20, DBLP:journals/dc/AlposCTZ24, DBLP:conf/podc/Amores-SesarCVZ25, DBLP:conf/opodis/Amores-SesarCKV25, DBLP:conf/wdag/LiCL23} and has been deployed in systems like the XRP Ledger and Stellar.

The flexibility of asymmetric trust comes at a cost. Algorithms for this model are typically more complex and difficult to design than their symmetric counterparts. In contrast, symmetric protocols have been studied for decades, and efficient solutions are known for most fundamental tasks.

Correct processes in asymmetric trust are characterized using the depth hierarchy~\cite{DBLP:conf/opodis/Amores-SesarCKV25}, a measure of the quality of their trust assumptions with respect to a failure scenario and which generalizes the binary distinction between faulty and correct processes in symmetric systems. Task solvability in the asymmetric setting is tightly connected to depth~\cite{DBLP:conf/opodis/Amores-SesarCKV25}. Given a failure scenario, for processes with depth zero no tasks can be solved. For processes with depth one, only simple tasks, such as consistent broadcast, are solvable, since this category restricts how much information can be exchanged in the protocol. In contrast, all more complex tasks can only be solved for processes with higher depths of at least two. This includes in particular reliable broadcast or consensus that form the basis of many large-scale distributed systems. The hierarchy that highlights this relationship is a key finding of this work~(Table~\ref{tab:depth-table}).

One of the main advantages of asymmetric trust is its ability to tolerate strictly more failure patterns than symmetric trust~\cite{DBLP:journals/dc/AlposCTZ24}. In the latter, if the actual failure pattern violates the global assumption (for example, having more than $f$ faults with a threshold failure assumption), correctness may fail for all processes. In asymmetric trust, the impact of such misguided trust assumptions is mostly local to the process that made them, but can still affect other processes. The relationship between this additional fault tolerance and task solvability is poorly understood. In particular, it is unclear if it is actually possible to solve problems in the extra fault tolerance cases, and if so, which tasks can be solved. In this paper, we study this exact relationship.

\paragraph*{Contribution.} We show that for every asymmetric quorum system, all failure scenarios that leave at least one process with depth two can also be tolerated by a suitably constructed symmetric quorum system. The latter is derived from the individual quorum systems of all processes. In the remaining scenarios tolerated only by the asymmetric system, all correct processes have depth at most one, making it impossible to solve tasks that require depth two or more. Thus, for such tasks, every solvable failure scenario under asymmetric trust is also solvable under symmetric trust, enabling the use of simpler symmetric algorithms. Our proof is constructive (Figure~\ref{fig:asym-fault-tolerance}): we present a compiler that transforms any asymmetric quorum system into a symmetric one with identical fault tolerance for all scenarios containing a process of depth at least two. The resulting quorum system can be efficiently used given the trust assumptions. This improves the best previously existing compiler, which achieved this only for processes with infinite depth~\cite{DBLP:conf/opodis/Amores-SesarCKV25}. We also show that no compiler exists that achieves the same for all scenarios with a process with depth one. 

Overall, although asymmetric trust tolerates more failure patterns, this extra tolerance does not increase solvability for tasks requiring depth two or more. For primitives such as reliable broadcast and consensus, symmetric trust is therefore sufficient.

\paragraph*{Organization.} The remainder of this paper is organized as follows. Section~\ref{sec:preliminaries} introduces the model and reviews the key concepts underlying our approach. Section~\ref{sec:asymmetric} formalizes the asymmetric trust model. Sections~\ref{sec:tolerance-asymmetric} presents our main result, the study of the relationship between fault tolerance and task solvability in asymmetric trust. Sections~\ref{sec:impact-known}~and~\ref{sec:impact-unknown} analyze our results in the two settings asymmetric trust. Section~\ref{sec:relatedwork} surveys related work on asymmetric trust. Finally, Section~\ref{sec:conclusions} presents the conclusions and points to different avenues for future work. 

\section{Preliminaries}
\label{sec:preliminaries}
\subsection{Model}
\label{sec:model}
All algorithms proposed in this paper operate in the asynchronous unauthenticated setting. We consider a system of $n$ \emph{processes} $\mathcal{P}=\{p_1, \dots, p_n\}$ that interact by exchanging messages. A protocol for $\mathcal{P}$ consists of a collection of programs, one for each process. We describe protocols using the event-based notation of Cachin~\etal~\cite{DBLP:books/daglib/0025983}.

A process that follows its protocol throughout an execution is called \emph{correct}. A \emph{faulty} process, also referred to as \emph{Byzantine}, may crash or deviate arbitrarily from its specification. We assume authenticated channels for message exchange between processes. For a system $\mathcal{A} \subseteq 2^{\mathcal{P}}$, the notation $\mathcal{A}^*$ denotes the collection of all subsets of the sets in $\mathcal{A}$, that is, $\mathcal{A}^* = \{A' \mid A' \subseteq A,\ A \in \mathcal{A}\}.$

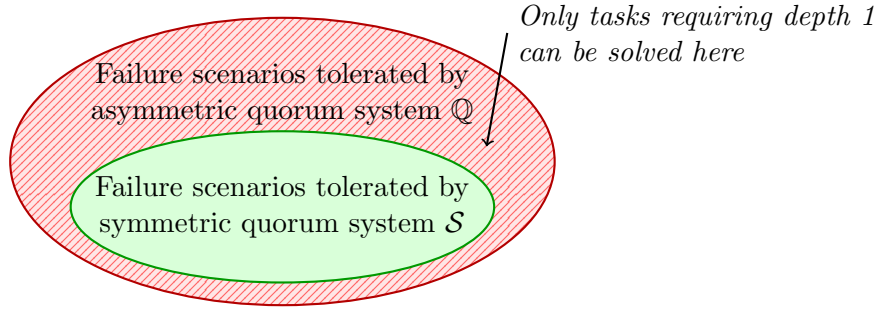
\begin{figure}
\centering
\begin{tikzpicture}

% Outer ellipse: Asymmetric quorum system
\fill[red!10] 
  (0,0) ellipse (3.6cm and 1.9cm);

\fill[pattern=north east lines, pattern color=red!60]
  (0,0) ellipse (3.6cm and 1.9cm);

\draw[red!70!black, thick]
  (0,0) ellipse (3.6cm and 1.9cm);

% Inner ellipse: Symmetric quorum system
\fill[green!15]
  (0,-0.6) ellipse (2.8cm and 1cm);

\draw[green!60!black, thick]
  (0,-0.6) ellipse (2.8cm and 1cm);

% Labels
\node[align=center] at (0,0.9) 
  {Failure scenarios tolerated by\\
   asymmetric quorum system $\mathbb{Q}$};

\node[align=center] at (0,-0.6) 
  {Failure scenarios tolerated by\\
   symmetric quorum system $\mathcal{S}$};

% Comment on the side
\node[align=left] (comment) at (5.5,1.7)
  {\emph{Only tasks requiring depth 1}\\
   \emph{can be solved here}};

\draw[->, thick] (comment.west) -- (2.7,0.2);

\end{tikzpicture}
\caption{The  set of failure scenarios tolerated by any asymmetric quorum system $\mathbb{Q}$ (in striped red) is strictly larger than that which can be tolerated by the symmetric quorum system $\mathcal{S}=\mathcal{C}(\mathbb{Q})$ derived by the compiler from $\mathbb{Q}$ (in green). However, in the failure patterns tolerated by $\mathbb{Q}$ but not $\mathcal{S}$, only tasks requiring depth one can be solved, this excludes problems like reliable broadcast or consensus.}
\label{fig:asym-fault-tolerance}
\end{figure}

\subsection{Symmetric trust overview}
\label{sec:symmetric-trust}
A symmetric fail-prone system is defined as a family of sets $\CF \subseteq 2^{\CP}$, where each $F \in \CF$ with $F \subseteq \CP$ is called a \emph{fail-prone set} and contains all processes that may at most fail together in some execution~\cite{DBLP:journals/dc/MalkhiR98}. They represent the assumption on the possible failure patterns that can occur and, as such, specify all sets of faulty processes that must be tolerated by a protocol. A protocol designed for \CF achieves its properties as long as the set $F$ of actually faulty processes satisfies $F \in \CF^*$. Symmetric quorum systems, which are specified with respect to a fail-prone system, are specified in Definition~\ref{def:quorum}.

\begin{definition}[(Symmetric) Byzantine quorum system~\cite{DBLP:journals/dc/MalkhiR98}]\label{def:quorum}
  A \emph{(symmetric) Byzantine quorum system} for \CF is a collection of sets of
  processes $\CQ \subseteq 2^{\CP}$ where no set is contained in
  another and each $Q \in \CQ$ is called a
  \emph{quorum}, such the following properties hold:
  \begin{description}
  \item[Consistency:] The intersection of any two quorums contains at least
    one process that is not faulty, i.e.,
    \[\forall Q_1, Q_2 \in \CQ , \forall F \in \CF: \, Q_1 \cap Q_2 \not
    \subseteq F.\]
\item[Availability:] For any set of processes that may fail together, there
  exists a disjoint quorum in \CQ, i.e.,
  \[\forall F \in \CF: \, \exists Q \in \CQ: \, F \cap Q = \emptyset.\]
  \end{description}
\end{definition}

This is a generalization of threshold failure assumptions for Byzantine faults~\cite{DBLP:journals/jacm/PeaseSL80}. In the same way in which many tasks can only be solved if $n > 3f$ in the threshold case, Byzantine quorum systems also require a bound on the faulty processes. This is captured through the $Q^3$ condition. 

\begin{definition}[$Q^3$-condition~\cite{DBLP:journals/dc/MalkhiR98,DBLP:journals/joc/HirtM00}]
  A fail-prone system \CF satisfies the \emph{$Q^3$-condition}, abbreviated
  as $Q^3(\CF)$, whenever it holds
  \[
    \forall F_1, F_2, F_3 \in \CF: \, \CP \not\subseteq F_1 \cup F_2 \cup F_3.
  \]
\end{definition}

The \emph{bijective complement} of a process set $\CS \subseteq 2^{\CP}$ is defined as $\overline{\CS} = \{ \CP \setminus S | S \in \CS \}$. The quorum system $\CQ = \overline{\CF}$ is called the \emph{canonical quorum system} of~\CF. Lemma~\ref{lem:canon} shows that a Byzantine quorum system exists if and only if the fail-prone system satisfies the $Q^3$ condition. 

\begin{lemma}[Byzantine quorum system existence~\cite{DBLP:journals/joc/HirtM00}]\label{lem:canon}
  Given a fail-prone system \CF, a Byzantine quorum system for \CF exists
  if and only if~$Q^3(\CF)$.
  In particular, if $Q^3(\CF)$ holds, then $\overline{\CF}$, the canonical quorum system associated to \CF, is a Byzantine quorum system.
\end{lemma}

\section{Asymmetric trust}
\label{sec:asymmetric}
We  consider the asymmetric trust model proposed by Alpos~\etal~\cite{DBLP:journals/dc/AlposCTZ24}.
We refer the reader to their paper for a full presentation. In protocols with asymmetric trust, each participant is free to make its own individual trust assumptions about others, captured by an asymmetric quorum system. Given a set of processes $\mathcal{P}$, an \emph{asymmetric fail-prone system} $\mathbb{F} = [\mathcal{F}_1, \dots, \mathcal{F}_n]$, where $\mathcal{F}_i$ represents the failure assumptions  of process $p_i$, captures the heterogeneous model. Each $\mathcal{F}_i$ is a collection of subsets of $\mathcal{P}$ such that some $F \in \mathcal{F}_i$ with $F \subseteq \mathcal{P}$ is called a fail-prone set for $p_i$ and contains all processes that, according to $p_i$, may at most fail together in some execution \cite{DBLP:conf/asiacrypt/DamgardDFN07}. We can, in turn, proceed to define asymmetric Byzantine quorum systems, denoted by $\mathbb{Q}$. 
\begin{definition}
\label{def:abqs}
  An \emph{asymmetric Byzantine quorum system} $\mathbb{Q}$ for $\mathds{F}$ is an array of collections of sets $\mathbb{Q} = [\mathcal{Q}_1, \cdots, \mathcal{Q}_n]$ where $\mathcal{Q}_i \subseteq 2^{\mathcal{P}}$ for $i \in [1, n]$. The set $\mathcal{Q}_i \subseteq 2^{\mathcal{P}}$ is a symmetric quorum system of $p_i$ and any set $Q_i \in \mathcal{Q}_i$ is called a quorum for $p_i$. The system $\mathbb{Q}$ must satisfy the following two properties.

\begin{description}
    \item[Consistency:] The intersection of two quorums for any two processes contains at least one process for which either process assumes that it is not faulty, i.e.,
    $$
    \forall i, j \in [1, n], \forall Q_i \in \mathcal{Q}_i, \forall Q_j \in \mathcal{Q}_j, \forall F_{ij} \in \mathcal{F}_i^* \cap \mathcal{F}_j^*: Q_i \cap Q_j \nsubseteq F_{ij}.
    $$

    \item[Availability:] For any process $p_i$ and any set of processes that may fail together according to $p_i$, there exists a disjoint quorum for $p_i$ in $\mathcal{Q}_i$, i.e.,
    $$
        \forall i \in [1, n], \forall F_i \in \mathcal{F}_i : \exists Q_i \in \mathcal{Q}_i : F_i \cap Q_i = \emptyset.
    $$
\end{description}    
\end{definition}

A kernel is a set of processes that intersects all quorums of a process. That is, for a process $p_i$ with quorum system $\mathcal{Q}_i$, $K$ is a kernel when $K \cap Q \neq \emptyset , \forall Q \in \mathcal{Q}_i$. We denote by $\mathcal{K}_i$ the set containing all kernels of a process $p_i$. 

Given an asymmetric fail-prone system $\mathbb{F}$, there will exist a valid asymmetric quorum system for $\mathbb{F}$ if and only if $\mathbb{F}$ satisfies the $B^3$ condition~\cite{DBLP:conf/asiacrypt/DamgardDFN07, DBLP:journals/dc/AlposCTZ24}. This property is defined as follows.
\begin{definition}[$B^3$-condition]
\label{def:b3}
  An asymmetric fail-prone system \BF satisfies the
  \emph{$B^3$-condition}, abbreviated as $B^3(\BF)$, whenever it holds that
  \[
    \forall i,j \in [1,n],
    \forall F_i \in \CF_i, \forall F_j\in\CF_j,
    \forall F_{ij} \in {\CF_i}^*\cap{\CF_j}^*: \,
    \CP \not\subseteq F_i \cup F_j \cup F_{ij} 
  \]
\end{definition}

If $B^3(\mathbb{F})$ holds, then the canonical quorum system, defined as the complement of the asymmetric fail-prone system, is a valid asymmetric quorum system. Throughout this work, we assume that all quorum systems are canonical, which simplifies the analysis; however, our results also extend to non-canonical systems. Accordingly, when we say that a quorum system satisfies the $B^3$ or $Q^3$ property, we mean that its associated canonical fail-prone system satisfies the property. 

In asymmetric systems, processes are still classified as correct or faulty, but correctness admits multiple levels depending on the accuracy of a process’s trust assumptions in a given execution. Let $F$ denote the set of faulty processes. This set is unknown to the processes and can be identified only by an external observer or the adversary. A process $p_i$ is said to correctly foresee $F$ if $F \in \mathcal{F}_i^*$, that is, if $F$ is contained in one of its fail-prone sets. Unlike in the symmetric setting, asymmetric systems may include processes that correctly foresee $F$ alongside others whose trust assumptions are incorrect and fail to capture the actual failures. This heterogeneity is captured by the notion of depth~\cite{DBLP:conf/opodis/Amores-SesarCKV25}. In contrast, symmetric systems implicitly assume that $F$ is always foreseen by the common quorum system. The depth of a process reflects its level of correctness and determines how much it can rely on other processes. We formalize this notion in Definition~\ref{def:depth}.

\begin{definition}[Depth of a process]
\label{def:depth}
    For an execution with faulty processes $F$, we recursively define the notion of a correct process having depth $d$ as follows
    \begin{itemize}

        \item Any correct process $p_i$ has depth 0. 
    
        \item Additionally, a correct process $p_i$ has depth $d\geq 1$ if it has a quorum such that all processes contained in it have depth at least $d-1$, i.e.,
        $$
        \exists Q \in \mathcal{Q}_i, \forall p_j \in Q: \text{ $p_j$ is correct, has depth $s$, and $s\geq d-1$}
        $$
    \end{itemize}
\end{definition}

A process with maximal depth $d$ also has depth $d'$ for all $0\leq d'\leq d$. We focus on the \emph{maximal depth} of a process. Processes with depth 0, although correct, cannot rely on any of their quorums, processes with depth $\infty$ can have absolute reliance on other processes. Note that in symmetric systems all processes have either depth $\infty$ (when trust assumptions hold) or 0 (otherwise). 

\subsection{Quorum knowledge}

Among the many works that study asymmetric trust, some~\cite{DBLP:journals/dc/AlposCTZ24, DBLP:conf/opodis/Amores-SesarCKV25, DBLP:conf/opodis/SheffWRM20} assume that each participant’s trust choices are globally known, while others~\cite{mazieres2015stellar, DBLP:conf/wdag/LosaGM19, DBLP:conf/wdag/LiCL23} assume that these choices are known only locally. Despite their prevalence, these settings have not been formally distinguished, and limited attention has been paid to their impact on solvability and algorithm design. We therefore distinguish between the known quorums and unknown quorums settings.

In the known quorums setting, each process knows the trust assumptions of every other process, as formalized in Definition~\ref{def:known-quorums}. This setting can be conceptualized via a trusted third party that collects and distributes quorum systems. Upon joining, each process $p_i$ submits its quorum system $\mathcal{Q}_i$ to this party, which verifies the $B^3$ property, combines the systems into a global quorum system $\mathbb{Q}$, and distributes it to all processes. The consensus protocols of Alpos~\etal~\cite{DBLP:journals/dc/AlposCTZ24} and Amores-Sesar~\etal~\cite{DBLP:conf/opodis/Amores-SesarCKV25, DBLP:conf/podc/Amores-SesarCVZ25} are designed for this setting. This setting applies in the permissioned, quasi-permissionless, and dynamically available models in the hierarchy of Lewis-Pye~and~Roughgarden~\cite{DBLP:journals/corr/abs-2304-14701}.

\begin{definition}[Known quorums setting]
    \label{def:known-quorums}
    Given an asymmetric quorum system $\mathbb{Q}$ and a set of processes $\mathcal{P}$, in the \textit{known-quorums} setting, each process $p_i \in \mathcal{P}$ has knowledge of $\mathbb{Q}$.
\end{definition}

In the unknown quorums setting, each process knows only its own quorum system $\mathcal{Q}_i$, as formalized in Definition~\ref{def:unknown-quorums}. This setting raises fundamental challenges for solvability. In particular, many problems require quorum systems to be pairwise intersecting, yet in this setting processes have no way of verifying whether this condition holds. A common approach in the literature~\cite{mazieres2015stellar, DBLP:conf/wdag/LosaGM19} is to assume pairwise intersection without providing guarantees when this assumption is violated. The reliable broadcast protocols of Alpos~\etal~\cite{DBLP:journals/dc/AlposCTZ24} and Amores-Sesar~\etal~\cite{DBLP:conf/opodis/Amores-SesarCKV25}, as well as the consensus protocol of Mazières~\cite{mazieres2015stellar}, exemplify this approach. Although techniques developed for the unknown quorums setting also apply to the known quorums setting, solutions tailored to the latter often achieve stronger guarantees and better efficiency. This setting applies in the fully permissionless model in the hierarchy of Lewis-Pye~and~Roughgarden~\cite{DBLP:journals/corr/abs-2304-14701}.

\begin{definition}[Unknown quorums setting]
    \label{def:unknown-quorums}
    Given an asymmetric quorum system $\mathbb{Q}$ and a set of processes $\mathcal{P}$, in the \textit{unknown-quorums} setting, each process $p_i \in \mathcal{P}$ only has knowledge of its own quorum system $\mathcal{Q}_i \in \mathbb{Q}$.
\end{definition}

\subsection{Characterizing tasks through depth}
\label{sec:depth-tasks}

Problems in distributed computing are typically characterized by abstractions that specify the safety and liveness properties required to hold in an execution~\cite{DBLP:books/daglib/0025983}. For example, a consensus protocol must satisfy agreement, validity, and termination. A protocol solves a task if every admissible execution of the protocol satisfies the task specification. If such a protocol exists, it is said that the task is solvable. In the symmetric setting, little attention is paid to which processes must satisfy these properties, as it is generally assumed that all correct processes do so. This is reasonable, since correct processes are not meaningfully distinguished from one another in this setting. This changes under asymmetric trust, where correct processes are no longer identical. Instead one can distinguish between naive and wise correct processes~\cite{DBLP:journals/dc/AlposCTZ24} and more generally organize them into a depth hierarchy~\cite{DBLP:conf/opodis/Amores-SesarCKV25}, based on the strength of their trust assumptions with respect to a given failure pattern. Their ability to satisfy task properties depends directly on their depth. As such, properties in asymmetric tasks incorporate depth as a parameter for the processes. As an example, consider the asymmetric consistent broadcast task presented in Definition~\ref{def:consistent-broadcast}.

\begin{definition}[Depth-characterized asymmetric Byzantine consistent broadcast]\label{def:consistent-broadcast}
  A protocol for \emph{asymmetric Byzantine consistent broadcast} with sender~$p_s$ that solves the task for depth $d$, shortened as \op{CB[$d$]}, defined through the events \op{dac-broadcast}$(m)$ and \op{dac-deliver}$(m)$ satisfies the following properties:
    \begin{itemize}
        \item  \textbf{Validity: } If a correct process $p_s$ \op{dac-broadcasts} a message $m$ then eventually all processes with depth $\var{d}$ \op{dac-deliver} $m$.
        \item \textbf{Consistency: }   If a process with depth $\var{d}$ \op{dac-delivers} $m$ and another process with depth $\var{d}$ \op{dac-delivers} $m'$, then $m=m'$.
        \item \textbf{Integrity: }  Every process with depth $\var{d}$ \op{dar-delivers} $m$ at most once. %Moreover, if $p_s$ is correct and the receiver has depth $sd$, then $m$ was previously broadcast by $p_s$
    \end{itemize}
\end{definition}

A protocol is said to solve a task with depth $d$ if, for all executions and failure scenarios having at least one process with depth $d$,  all processes having depth $d$ satisfy the task properties at the end of the execution. If such a protocol exists, one can also say that the task is solvable with depth $d$. This raises the question of the minimum depth for which there exists a protocol that solves the task. We call this the minimum depth required to solve a task. Finding such depths for different tasks is an open research problem. 

For consensus and reliable broadcast, the minimum known depths for which there are protocols that solve the problem are ten and three respectively~\cite{DBLP:conf/opodis/Amores-SesarCKV25}. Decreasing such depths to the lowest values possible is an open research question. It has been shown that there exist no protocols that solve such tasks for depth one~\cite{DBLP:conf/opodis/Amores-SesarCKV25}. Therefore, the minimum required depth for such problems is at least two, but it could be higher. Similarly, it has been proven that there are no tasks that can be solved for depth 0, as a consequence, all asymmetric tasks require depth at least one~\cite{DBLP:conf/opodis/Amores-SesarCKV25}. At depth one, only simple tasks requiring minimal inter-process communication can be solved. Consistent broadcast is the main example of such a task. Note that if a task requires depth $d$, but a failure scenario only leaves processes having depths $d'<d$, the task is said to not be solvable for such scenario, as no process is guaranteed to satisfy all the properties required by the task definition. Table~\ref{tab:depth-table} summarizes the current characterization of tasks by minimum depths for which a protocol exists, in both quorum knowledge settings.

\begin{table}[ht]
\centering
\begin{tabular}{|c|c|c|}
\hline
\textbf{Depth} & \textbf{Known quorums setting} & \textbf{Unnown quorums setting} \\ \hline
$\infty$ &  & Partially-synchronous consensus~\cite{mazieres2015stellar} \\ \hline % Best known quorum compiler~\cite{DBLP:conf/opodis/Amores-SesarCKV25}
$\vdots$ &  & \\ \hline
10 & Asynchronous consensus~\cite{DBLP:conf/opodis/Amores-SesarCKV25}  & \\ \hline
$\vdots$ &  & \\ \hline
3 & Reliable broadcast~\cite{DBLP:conf/opodis/Amores-SesarCKV25}  & Reliable broadcast~\cite{DBLP:conf/opodis/Amores-SesarCKV25} \\ \hline
2 &   & \\ \hline % Our quorum compiler
1 \cellcolor{lightorange}&  Consistent broadcast~\cite{DBLP:journals/dc/AlposCTZ24} \cellcolor{lightorange} & Consistent broadcast~\cite{DBLP:journals/dc/AlposCTZ24} \cellcolor{lightorange}\\ \hline
0 \cellcolor{lightred} & \cellcolor{lightred} & \cellcolor{lightred} \\ \hline
\end{tabular}
\caption{Characterization of tasks by lowest  depth for which there is a protocol that solves it, in both settings of asymmetric trust. The red highlighted depth (0) shows the space where no tasks are solvable, the orange highlighted depth (1) shows the space  where only simple tasks, like consistent broadcast, can be solved. Many core distributed tasks such as reliable broadcast and consensus require depth at least two.}\vspace*{-\baselineskip}
\label{tab:depth-table}
\end{table}

\section{On the fault tolerance of asymmetric trust}
\label{sec:tolerance-asymmetric}
One of the main advantages of asymmetric trust is its greater fault tolerance compared to symmetric trust~\cite{DBLP:conf/asiacrypt/DamgardDFN07}. Its complex and heterogeneous structure enables the tolerance of failure patterns that cannot be handled by standard symmetric quorum systems. Consider the asymmetric quorum system presented in Example~\ref{ex:1}. This trust structure tolerates a wider range of failure patterns than any symmetric quorum system could handle (since the $Q^3$ property is not satisfied). 

In particular, this is achieved by having heterogeneous notions of fault-tolerance. For a given failure pattern, some processes might tolerate it while others might not. This contrasts with symmetric systems, where either all processes tolerate the failure scenario or none does. In this work, we investigate which tasks remain solvable in these additional scenarios.

Our analysis relies on the notion of \emph{depth} and its relationship to task solvability, introduced in Section~\ref{sec:depth-tasks}. It is known that no tasks are solvable for depth zero, and that only simple tasks  can be solved for depth one. In contrast, many core problems in distributed computing, such as reliable broadcast and consensus, require depth at least two. Since each failure scenario determines the depths of the correct processes, this characterization enables a simple and systematic study of task solvability under different fault patterns.

Specifically, if a fault pattern tolerated by an asymmetric quorum system leaves only processes with depth zero or one, then only a very restricted class of tasks can be solved, those requiring depth one. If, however, at least one process has depth two or higher, the system enters a regime in which more complex tasks become solvable, tasks requiring depth two or more. Even though a task requiring depth $d>2$ can't be solved if the maximum depth of a process is two, in all scenarios where it is solvable (i.e., there is at least one process with depth $d$), there will be at least a process with depth two (since a process with depth $d$ also has depth $d'$ for all $d'<d$). We refer to the scenario where the maximum depth of any process is one or less as the \emph{simple} scenario and to the latter as the \emph{complex} scenario. We study how the fault tolerance of asymmetric trust differs significantly between those two scenarios.

We show that the additional fault tolerance provided by asymmetric trust lies entirely within the simple scenario. In contrast, all fault patterns under which complex tasks remain solvable can also be tolerated by an appropriate symmetric quorum system.

\begin{example}
\label{ex:1}
Consider a system with five processes $\mathcal{P} = \{p_1, p_2, p_3, p_4, p_5\}$ and asymmetric quorum system $\mathbb{Q}$. This is a valid asymmetric quorum system as it satisfies the $B^3$ property. 
\[
\begin{aligned}
\mathcal{Q}_1 &= \bigl\{\{p_1, p_2\}, \{p_1, p_2, p_3, p_5\} \bigr\} \\
\mathcal{Q}_2 &= \bigl\{\{p_2, p_3, p_4\}, \{p_1, p_2, p_3, p_5\} \bigr\} \\
\mathbb{Q}:\quad
\mathcal{Q}_3 &= \bigl\{\{p_1, p_3, p_4\}, \{p_1, p_2, p_3, p_5\} \bigr\} \\
\mathcal{Q}_4 &= \bigl\{\{p_2, p_3, p_4, p_5\}, \{p_1, p_2, p_3, p_5\} \bigr\} \\
\mathcal{Q}_5 &= \bigl\{\{p_2, p_3, p_4, p_5\}, \{p_1, p_2, p_3, p_5\} \bigr\} \\
\end{aligned}
\]

The quorum system $\mathbb{Q}$ can tolerate more faults than any symmetric quorum system, as it does not satisfy the $Q^3$ property. For example, failure patterns $F=\{p_3, p_4, p_5\}$, $F=\{p_1\}$, $F=\{p_2\}$, $F=\{p_3\}$, $F=\{p_4\}$, or $F=\{p_5\}$ can be tolerated, something not possible for any symmetric quorum system. However, this extra expressive power lies in fault scenarios in which only tasks requiring depth 1 can be solved. When $F=\{p_3, p_4, p_5\}$, process $p_2$ has depth 0 and process $p_1$ has depth 1. Since the maximum depth of a process is 1, it is not possible to solve problems such as reliable broadcast or consensus. Note that all failure scenarios that allow tasks requiring depth two to be solved ($F=\{p_4\}$, $F=\{p_5\}$) could be handled with a 4-out-of-5 symmetric threshold quorum system, in fact, such symmetric quorum system would be able to guarantee solvability of problems for more  failure scenarios than $\mathbb{Q}$. We generalize this in Lemmas~\ref{lem:asym-collapse-d2}~and~\ref{lem:d2-compiler-valid}, where we show that, for any asymmetric quorum system, all fault patterns that allow tasks requiring depth two to be solved with asymmetric trust could also be handled with symmetric trust. 
    
\end{example}

More precisely, we show that for any asymmetric quorum system~$\mathbb{Q}$, there exists a symmetric quorum system~$\mathcal{S}$ such that, for every failure pattern tolerated by~$\mathbb{Q}$ in which complex problems can be solved,~$\mathcal{S}$ can also tolerate that pattern and support the solution of the same tasks. Consequently, when attention is restricted to scenarios in which complex tasks are solvable, symmetric and asymmetric trust provide identical fault-tolerance guarantees.

In Example~\ref{ex:1}, the system~$\mathbb{Q}$ tolerates the failure pattern $F={p_3, p_4, p_5}$. However, these failures cause all correct processes to have depth zero or one, and therefore tasks requiring higher depths can't be solved. 

We establish this result by developing a \emph{quorum compiler} that transforms any asymmetric quorum system into a symmetric quorum system. We show that, in every execution in which at least one process has depth at least two, the resulting symmetric quorum system satisfies the standard consistency and availability properties. This guarantees that it is a valid symmetric quorum system and that it can be used within symmetric algorithms to solve distributed tasks. 

We also show that there exist no compilers that guarantee that the symmetric quorum system generated satisfies the consistency and availability properties if the maximum depth of any process is one. Together, these results provide a tight characterization of the power and limitations of compilers with respect to depth.

\subsection{Compiler}
\label{sec:compiler-d2}

In this section we study the problem of translating trust assumptions from the asymmetric to the symmetric world. We first introduce the notion of a set of failures induced by a depth for a particular asymmetric quorum system (Definition~\ref{def:di-failures}). That is, a set of failures such that if any of those occur, the system will have one or more processes with the inducing depth.

\begin{definition}
    \label{def:di-failures}
    Given an asymmetric quorum system $\mathbb{Q}$, denote by $\mathcal{F}[d] \subseteq 2^{\mathcal{P}}$ the set of all failure patterns that allow at least one process to have depth $d$ or more. That is, a set $F\subseteq\mathcal{P}$ belongs to $\mathcal{F}[d]$ if and only if there is a process $p \in \mathcal{P}$ such that when the parties in $F$ are faulty, $\text{depth}(p)\geq d$.
\end{definition}

Having defined $\mathcal{F}[d]$, we define the notion of a quorum compiler parametrized by depth. This is a program that translates asymmetric quorum systems into symmetric counterparts while keeping the same fault-tolerance guaranties if there is at least one process with a certain depth.

\begin{definition}[Depth-parametrized quorum compiler]
    A quorum compiler with depth $d$, shortened as $\mathcal{C}[d]$, is a program that takes as input any asymmetric quorum $\mathbb{Q}$ that satisfies the $B^3$ property and produces as output a symmetric quorum system $\mathcal{S}$, such that for all failure patterns in $\mathcal{F}[d]$, the resulting symmetric quorum system $\mathcal{S}$ satisfies the following properties
    \begin{enumerate}
        \item $\forall\ Q_1, Q_2 \in \mathcal{S}, \forall F \in \mathcal{F}[d]:  Q_1\cap Q_2 \nsubseteq F$
        \item $\forall\ F \in \mathcal{F}[d], \exists Q \in \mathcal{S}: Q\cap F = \emptyset$
    \end{enumerate}
    That is, $\mathcal{S}$ satisfies the quorum consistency and availability properties when considering $\mathcal{F}[d]$ as the fail-prone system whose faults have to be tolerated. 
\end{definition}

We parametrize the compiler using depth since this is the best existing tool to characterize asymmetric trust. Note, however, that there could exist other parameterizations that allow even more expressive power at the time of characterizing processes in asymmetric trust. 

We now show a compiler construction that implements $\mathcal{C}[2]$. Before this, the best existing construction~\cite{DBLP:conf/opodis/Amores-SesarCKV25} implemented $\mathcal{C}[\infty]$. Afterwards, we show that it is not possible to construct a compiler that implements $\mathcal{C}[1]$. This shows that it is not possible to get a compiler with a lower depth parameter than the one presented here.   Algorithm~\ref{alg:compiler} presents the code to implement the compiler $\mathcal{DC}$.  

\begin{algo*}[h!]
\vbox{
\small
\begin{numbertabbing}\reset
  xxxx\=xxxx\=xxxx\=xxxx\=xxxx\=xxxx\=MMMMMMMMMMMMMMMMMMM\=\kill
  $\mathcal{S} \gets \emptyset$ // stores the symmetric quorum system \label{}\\
  \textbf{for each} $\mathcal{Q}_i \in \mathbb{Q}$ \textbf{do} \label{line:originating_quorum}\\
  \> \textbf{for each} $Q_i \in \mathcal{Q}_i$ \textbf{do }\label{}\\
  \>\> expand($Q_i$, 0, $\emptyset$)\label{line:expand}\\
  \textbf{return} $\mathcal{S}$ \label{}\\ \\

  \textbf{function} expand($Q$, $j$, \text{currentquorum}) \label{}\\
  \> \textbf{if} $j < |Q|$ \textbf{do} \label{}\\
  \>\> $p_k \gets Q[j]$ // process at index $j$ in $Q$ \label{}\\
  \>\> \textbf{for each } $Q' \in \mathcal{Q}_k\textbf{ do}$ \label{}\\
  \>\>\> expand($Q$, $j+1$, $\text{currentquorum} \cup Q'$) \label{}\\
  \> \textbf{else} \label{}\\
  \>\> $\mathcal{S} \gets \mathcal{S} \cup \text{currentquorum}$ \label{}
\end{numbertabbing}
}
\caption{Asymmetric to symmetric quorum compiler $\mathcal{DC}$}
\label{alg:compiler}
\end{algo*}

\begin{definition}[Compiler $\mathcal{DC}$]
\label{def:s}
    Given any $B^3$-satisfying asymmetric quorum system $\mathbb{Q} = [\mathcal{Q}_1, \dots, \mathcal{Q}_n]$, we define the compiler $\mathcal{DC}$ as follows. It takes as input $\mathbb{Q}$ and outputs a symmetric quorum system $\mathcal{S} = \mathcal{DC}(\mathbb{Q})$.
    
    The compiler transforms an asymmetric quorum system into a symmetric one by expanding each quorum. More precisely, for every quorum $Q \in \mathcal{Q}_i \in \mathbb{Q}$ in the local quorum systems of all processes, the compiler considers each process $p \in Q$ and replaces it with one of its local quorums from $\mathcal{Q}_p$. It enumerates all possible combinations of such choices (i.e., the Cartesian product), takes the union of the selected sets, and collects every resulting union into the output system $\mathcal{S}$.

    %Mathematically this can be formalized as follows: 

    %\mathcal{S} = \bigl\{ \bigcup \mathcal{Q}_1 \times \mathcal{Q}_2 \times \dots \mathcal{Q}_k \mid p_1, \dots, p_k \in Q, \forall Q \in \mathcal{Q}_i \forall p_i \in \mathcal{P} \bigr\}$.

\end{definition}
Intuitively, each original quorum is "expanded" by substituting every member with one of its local quorums, and the symmetric quorum system consists of all sets that can be formed this way.

Given a symmetric quorum $S \in \mathcal{S}$, we denote by $O(S)$ the original quorum that was “expanded” to produce $S$, i.e., the quorum $Q_i$ in Line~\ref{line:expand} of the algorithm. From $\mathcal{S}$, we can also define its associated canonical fail-prone system $\mathcal{N} = \{\mathcal{P} \setminus S \mid S \in \mathcal{S}\}$. Lemma~\ref{lem:asym-collapse-d2} shows that if $\mathbb{Q}$ satisfies $B^3$, then $\mathcal{S}$ satisfies $Q^3$. Lemma~\ref{lem:d2-compiler-valid} shows that if a process has depth at least 2, $\mathcal{S}$ satisfies quorum consistency and availability. Together, these results show that $\mathcal{C}$ produces a valid symmetric quorum system.

\begin{restatable}[]{lemma}{lone}
    \label{lem:asym-collapse-d2}
    Let $\mathbb{Q}$ be an asymmetric Byzantine quorum system among processes $\mathcal{P}$ with asymmetric fail-prone system $\mathbb{F}$ and let $\mathcal{S} = \mathcal{DC}(\mathbb{Q})$ be  the symmetric Byzantine quorum system obtained with compiler $\mathcal{DC}$. Let $\mathcal{N}$ be the canonical fail-prone system associated to $\mathcal{S}$. If $B^3(\mathbb{F})$ then $Q^3(\mathcal{N})$. 
\end{restatable}
\begin{proof}
    See Appendix~\ref{ap:proofs}.
\end{proof}

\begin{restatable}[]{lemma}{ltwo}
    \label{lem:d2-compiler-valid}
    The compiler $\mathcal{DC}$ implements $\mathcal{C}[2]$
    %Given an asymmetric quorum system $\mathbb{Q}$ and a set of processes $\mathcal{P}$, if there is at least one process with depth two, then $\mathcal{S} = \mathcal{DC}(\mathbb{Q})$ implements $\mathcal{C}[2]$. satisfies the consistency and availability properties
\end{restatable}
\begin{proof}
    See Appendix~\ref{ap:proofs}.
\end{proof}

\begin{restatable}[]{lemma}{lthree}
    Consider an execution with faulty processes $F$, a process $p_i$ such that $\op{depth(}p_i\op{)} \geq 1$ and any of its quorums $Q_i \in \mathcal{Q}_i$. For any other process $p_j$ with $\op{depth(}p_j\op{)} \geq 1, \exists K_j \in \mathcal{K}_j$ such that $K_j \subseteq Q_i \setminus F$.
    \label{lem:q-always-k}
\end{restatable}
\begin{proof}
    See Appendix~\ref{ap:proofs}.
\end{proof}

\begin{corollary}
    Consider an execution with faulty processes $F$, and two processes $p_i$, $p_j$ with depth at least one. For any pair of their quorums $Q_i \in \mathcal{Q}_i$, $Q_j \in \mathcal{Q}_j$, $Q_i \cap Q_j \nsubseteq F$. 
    \label{lem:d1-intersect}
\end{corollary}
%\begin{proof}
%    From Lemma~\ref{lem:q-always-k} it follows that given any quorum $Q_i \in \mathcal{Q}_i$ of process $p_i$, there is a set composed solely of correct processes that intersects all quorums of $p_j$. This proves the theorem. 
%\end{proof}

This result implies that all failure patterns contemplated by $\mathcal{F}[2]$ are tolerated by the symmetric quorum system. However, if no process has depth two or more, the compiler does not guarantee consistency or availability. 
For tasks requiring depth two or more, no asymmetric protocol could provide a solution in these cases anyway. This is the reason why every failure scenario that permits solving a task requiring depth at least two can also be tolerated by the symmetric quorum system produced by the compiler.

\subsection{Impossibility of compilers for depth one}
\label{sec:compilers-impossible}
In Section~\ref{sec:compiler-d2} a compiler from an asymmetric quorum system $\mathbb{Q}$ to a symmetric quorum system $\mathcal{S}$ with depth parameter two was presented, denoted as $\mathcal{C}[2]$. It maintains the same fault-tolerance guarantees for all failure patterns that leave at least one process with depth two. In this section we show that there cannot exist a compiler that implements $\mathcal{C}[1]$. This shows a tight bound on the existence of compilers parametrized by depth. Lemma~\ref{lem:d1-compiler-impossible} proves this result. 

\begin{restatable}[]{lemma}{lfour}
    \label{lem:d1-compiler-impossible}
    There exists no compiler that implements $\mathcal{C}[1]$.
\end{restatable}
\begin{proof}
    See Appendix~\ref{ap:proofs}.
\end{proof}

As discussed in Section~\ref{sec:compiler-d2}, our results do not rule out every possible characterization of compilers using other concepts rather than depth. Depth is the most expressive concept for asymmetric quorum systems; nonetheless, we leave a complete characterization of compilers as future work.

\section{Impact on the known-quorums setting}
\label{sec:impact-known}

In the previous section, we showed that for any asymmetric quorum system $\mathbb{Q}$, there exists a symmetric quorum system $\mathcal{S}$ that tolerates the same failures in scenarios where tasks requiring depth two or more are solvable. For fault scenarios tolerated by $\mathbb{Q}$ but not $\mathcal{S}$, only tasks requiring depth one can be solved. This existential result applies to both the known and unknown quorums settings, though its practical application differs between the two.

In the known quorums setting, all processes have full knowledge of $\mathbb{Q}$, allowing each to locally compute $\mathcal{S} = \mathcal{C}(\mathbb{Q})$. They can then execute any symmetric protocol $\mathcal{A}$ with $\mathcal{S}$ to solve tasks, avoiding more complex asymmetric protocols. For tasks requiring depth two or more, every execution of a symmetric protocol with $\mathcal{S}$ allows each correct process to obtain a correct output from the execution. Processes with depths zero and one benefit most, since an asymmetric protocol for depth two would not guarantee a solution for them. Thus, using the compiled symmetric quorum system can offer stronger guarantees than asymmetric protocols. The cost of building and using this system is discussed in Section~\ref{sec:compiler-cost}.

\subsection{Cost of using the compiler}
\label{sec:compiler-cost}

A quorum system can in principle describe the power set of all processes and thus be exponentially large in $n$. This mean that in some scenarios, checking if a certain message has been received from a quorum could already be computationally expensive.
Thus explicitly executing the compiler in Algorithm~\ref{alg:compiler} to construct the symmetric quorum system may be impractical, as it incurs quadratic complexity in the size of the original asymmetric system.

Fortunately, checking whether a set contains a quorum can be done efficiently without explicitly constructing the full system. An on-the-fly procedure can determine if a set of received messages contains a quorum produced by the compiler, incurring only an additional $O(n)$ overhead compared to checking quorum inclusion in the original asymmetric system.
We note that this linear computational overhead is unavoidable in a compiled system since the original system may rely critically on the assumptions of one particular process (e.g. the only process with a certain depth), and the compiler must work regardless of which processes has maximal depth.\footnote{If the compiler is allowed to guess a process with maximal depth, a trivial solution is to simply output the quorums of that process.}

Recall that the compiler constructs symmetric quorums by expanding a single asymmetric quorum: starting from a quorum of one process, it selects, for each member, one of that member’s own quorums. A set of received messages contains a symmetric quorum if it contains this two-level structure. Suppose a process has received messages from a set $M$. The check proceeds in two steps:
\begin{description}
    \item[Filter parties:] For each party $p \in \mathcal{P}$, check whether $p$ has a quorum fully contained in $M$. Parties that fail this check cannot belong to any symmetric quorum contained in $M$. Let $M' \subseteq M$ denote the parties that pass this check.
    \item[Identify originating quorum:] Check whether there exists a party $p \in \mathcal{P}$ with a quorum $Q \in \mathcal{Q}_i$, such that $Q \subseteq M'$. If so, $M$ contains a quorum generated by the compiler, which resulted from expanding asymmetric $Q$.
\end{description}

Each step requires checking quorum inclusion for every party. Hence, the total cost of determining whether a symmetric quorum exists is at most $2n$ quorum checks in the original asymmetric system.

\section{Impact in the unknown quorums setting}
\label{sec:impact-unknown}
In the unknown quorums setting, each process $p_i$ does not know the complete asymmetric quorum system $\mathbb{Q}$, only their own trust assumptions $\mathcal{Q}_i$. It is possible for processes to exchange trust assumptions between them; however, since we are in the Byzantine setting, processes could lie or send conflicting information about their own trust assumptions. Additionally, processes can report local quorums that cause the $B^3$ property to break. Combined, this makes it very difficult for processes to learn or agree on the same asymmetric quorum system $\mathbb{Q}$. Consequently, using the compiler to produce a symmetric quorum system is not immediately possible. Lemmas~\ref{lem:asym-collapse-d2}~and~\ref{lem:d2-compiler-valid} still guarantee the existence of a symmetric quorum system with the same fault tolerance for cases where tasks requiring depth two or more can be solved, however, constructing it is not straightforward. We leave as an open question whether it is possible to apply ideas like the ones presented in this work also in the unknown-quorums setting.

\section{Related work}
\label{sec:relatedwork}

Thresholds and symmetric quorum systems are fundamental tools in distributed computing, widely used to guarantee correctness in protocols~\cite{DBLP:journals/siamcomp/NaorW98,DBLP:journals/dc/MalkhiR98,DBLP:books/daglib/0017536,DBLP:books/daglib/0025983}. They are also core components of many practical systems, including cloud platforms~\cite{DBLP:journals/sigops/LakshmanM10, DBLP:conf/usenix/HuntKJR10} and cryptocurrencies~\cite{DBLP:conf/sosp/GiladHMVZ17, DBLP:journals/corr/abs-1807-04938, buterin2018ethereum2}. A key limitation is that all participants must rely on the same quorum system, which limits the expression of heterogeneous trust assumptions.

The asymmetric trust model, introduced by Damg{\aa}rd~\etal~\cite{DBLP:conf/asiacrypt/DamgardDFN07} and further developed by Alpos~\etal~\cite{DBLP:journals/dc/AlposCTZ24}, generalizes the symmetric paradigm by allowing participants to make independent trust choices. These choices can reflect social relationships or external information unavailable to the protocol. In this model, each participant defines its own quorum system, operating under individualized trust assumptions. Amores-Sesar~\etal~\cite{DBLP:conf/podc/Amores-SesarCVZ25} show that algorithms designed for symmetric quorum systems cannot, in general, be applied directly to the asymmetric setting, as key quorum properties may no longer hold. This motivates adapting symmetric protocols and redefining classical correctness properties for asymmetric trust.

A growing body of work has focused on designing primitives and protocols specifically for the asymmetric model. Alpos~\etal~\cite{DBLP:journals/dc/AlposCTZ24} introduce asymmetric variants of fundamental building blocks, including reliable broadcast, binary consensus, and common coin. Building on these primitives, Amores-Sesar~\etal~\cite{DBLP:conf/podc/Amores-SesarCVZ25} develop a DAG-based consensus protocol tailored to asymmetric trust. Sheff~\etal~\cite{DBLP:conf/opodis/SheffWRM20} propose a variant of Paxos that incorporates heterogeneous trust assumptions. Losa~\etal~\cite{DBLP:conf/wdag/LosaGM19} propose an alternative modeling approach that replaces fail-prone sets with strengthened quorum definitions, requiring each quorum to contain a quorum for each of its members. Li~\etal~\cite{DBLP:conf/wdag/LiCL23} extend this framework by identifying quorum properties necessary or sufficient to solve fundamental asymmetric problems such as consensus. Finally, Amores-Sesar~\etal~\cite{DBLP:conf/opodis/Amores-SesarCKV25} refine these results, showing that reliable broadcast and consensus can be solved under significantly weaker assumptions than previously thought.

Asymmetric consensus protocols have also been deployed in blockchain systems, most notably the XRP Ledger\footnote{\url{https://xrpl.org}}~\cite{DBLP:journals/corr/abs-1802-07242}
 and Stellar\footnote{\url{https://stellar.org}}~\cite{mazieres2015stellar,DBLP:conf/sosp/LokhavaLMHBGJMM19}. In the XRP Ledger, each participant specifies its trust assumptions by listing nodes whose votes it considers~\cite{DBLP:journals/corr/abs-1802-07242,DBLP:conf/opodis/Amores-SesarCM20}. Stellar follows a similar approach, with each participant maintaining a set of trusted nodes and waiting for a sufficient majority to agree on a transaction before finalizing it~\cite{mazieres2015stellar,DBLP:conf/wdag/LosaGM19,DBLP:conf/sosp/LokhavaLMHBGJMM19}.

Translating problems between computational models has a long history in distributed computing~\cite{DBLP:journals/ipl/GoldreichP90,DBLP:conf/eurocrypt/DeligiosE24, DBLP:journals/iacr/LossM18}. In the context of asymmetric trust, Senn~and~Cachin~\cite{DBLP:conf/papoc/SennC25} were the first to study quorum compilers, focusing on transformations between asymmetric and symmetric trust in the crash-fault models. In the Byzantine setting, Amores-Sesar~\etal~\cite{DBLP:conf/podc/Amores-SesarCVZ25} proposed a quorum compiler that transforms asymmetric quorum systems into symmetric ones, though the resulting system guarantees correctness only under very restrictive operating conditions, namely, under the existence of processes with infinite depth.

\section{Conclusion and future work}
\label{sec:conclusions}

This paper investigated task solvability in failure scenarios that are uniquely tolerated by asymmetric trust, through the lens of the depth hierarchy. Our main result shows that asymmetric trust does not allow to solve tasks requiring depths at least two, such as reliable broadcast and consensus, in more failure scenarios than symmetric trust does. For any asymmetric quorum system, every failure scenario that permits solving a task requiring depth at least two can also be tolerated by a suitably constructed Byzantine symmetric quorum system. In the remaining failure scenarios, which are uniquely tolerable by asymmetric trust, it would not be possible to solve such tasks anyway, since all processes have depth at most one. We establish this by presenting a compiler that transforms asymmetric quorums into symmetric ones while preserving fault tolerance.

In contrast, depth one tasks emerge as the primary beneficiaries of asymmetric trust. They are the only tasks that can be solved in failure scenarios that are uniquely supported by asymmetric quorum systems, highlighting a precise boundary in the hierarchy. Consistent broadcast, one of these tasks, has been proposed and deployed as a consensusless building block in several systems~\cite{DBLP:conf/aft/BaudetDS20, DBLP:conf/wdag/AlposDMSZ25, DBLP:conf/ccs/BlackshearCDKKL24}, demonstrating that, despite its simplicity, it remains a useful component. Incorporating it into asymmetric systems would enable them to fully exploit the additional fault tolerance that asymmetric quorum systems provide.

Our results show that, for many tasks, equivalent symmetric trust assumptions exist; however, constructing and deploying them efficiently remains an open question. 
In the known-quorums setting, while the compiler is optimal in terms of depth, it introduces a linear computational overhead which may be relevant in complex systems. One could hope to eliminate this overhead using a natively asymmetric protocol, for which the problem of achieving optimal depth remains open.
In the unknown-quorums setting, constructing a symmetric quorum system from an asymmetric one remains open.

\section*{Acknowledgments}

This work was supported by the Swiss National Science Foundation (SNSF)
under grant agreement Nr\@.~219403 (Emerging Consensus), by the Initiative for
Cryptocurrencies and Contracts (IC3), by the Cryptographic Foundations for
Digital Society, CryptoDigi, DFF Research Project 2, Grant ID
10.46540/3103-00077B, and by the European Union, ERC2023-StG-101116713. Views and opinions expressed are those of the author(s) 
only and do not necessarily reflect those of the European Union. Neither the European Union nor the granting authority can be held responsible for them.

\bibliographystyle{plainurl}% the mandatory bibstyle
%\clearpage
\bibliography{references, dblpbibtex}

\newpage
\appendix
\section{Proofs}
\label{ap:proofs}

\lone*
\begin{proof}
    We will proceed by contradiction. Assume that $\mathcal{N}$ does not satisfy the $Q^3$-condition. Therefore, there exist $N_1, N_2, N_3 \in \mathcal{N}$ such that $N_1 \cup N_2 \cup N_3 = \mathcal{P}$. We will denote by $S_1 = \mathcal{P}\setminus N_1, S_2= \mathcal{P}\setminus N_2, S_3= \mathcal{P}\setminus N_3$ the corresponding  canonical quorums.

    Denote by also $Q_1 = O(S_1), \ Q_2 = O(S_2)$, and $Q_3 = O(S_3)$. 
    By the asymmetric quorum consistency property, there are processes $p_i, p_j$ such that $p_i \in Q_1 \cap Q_2$ and $p_j \in Q_2 \cap Q_3$. 

    By the definition of $\mathcal{S}$, there exists a quorum $Q_i \in \mathcal{Q}_i$ such that $Q_i \subseteq S_1$ and there is another quorum $Q_i' \in \mathcal{Q}_i$ such that $Q_i' \subseteq S_2$. Denote by $F_i = \mathcal{P} \setminus Q_i$ and $F_i'= \mathcal{P} \setminus Q_i'$ the canonical asymmetric fail-prone sets associated to $Q_i, Q_i'$. Since $F_i = \mathcal{P} \setminus Q_i$, $Q_i \subseteq S_1$, and $N_1 = \mathcal{P} \setminus S_1$, it follows that $N_1 \subseteq F_i$. Reasoning the same logic we obtain that $N_2 \subseteq F_i$. Therefore, $N_1 \in \mathcal{F}_i$ and $N_2 \in \mathcal{F}_i$.

     By the definition of $\mathcal{S}$, there is another quorum $Q_j \in \mathcal{Q}_j$ such that $Q_j \subseteq S_2$ and there is a quorum $Q_j' \in \mathcal{Q}_j$ such that $Q_j' \subseteq S_3$. Denote by $F_j = \mathcal{P} \setminus Q_j$ and $F_j'= \mathcal{P} \setminus Q_j'$ the canonical asymmetric fail-prone sets associated to $Q_j, Q_j'$. Since $F_j = \mathcal{P} \setminus Q_j$, $Q_j \subseteq S_2$, and $N_2 = \mathcal{P} \setminus S_2$, it follows that $N_2 \subseteq F_j$. Applying the same reasoning we get $N_3 \subseteq F_j$. Therefore, $N_2 \in \mathcal{F}_j$ and $N_3 \in \mathcal{F}_j$.

    This constitutes the contradiction since $p_i$ with fail-prone sets $N_1$, $N_2$ and $p_j$ with fail-prone sets $N_2$, $N_3$ violate the $B^3$-condition in $\mathbb{Q}$ as $N_1 \cup N_3 \cup (N_2 \cap N_2) = N_1 \cup N_3 \cup N_2 = \mathcal{P}$.
\end{proof}

\ltwo*
\begin{proof}

    We want to show that the compiler $\mathcal{DC}$ implements $\mathcal{C}[2]$. Consider a set of processes $\mathcal{P}$ and any $B^3$-satisfying asymmetric quorum system $\mathbb{Q}$ and its associated canonical fail-prone system $\mathbb{F}$. To prove the lemma we must show that $\mathcal{S} = \mathcal{DC}(\mathbb{Q})$ satisfies the following properties

    \begin{enumerate}
        \item $\forall Q_1, Q_2 \in \mathcal{S}, \forall F \in \mathcal{F}[2]:  Q_1\cap Q_2 \nsubseteq F$
        \item $\forall F \in \mathcal{F}[2], \exists Q \in \mathcal{S}: Q\cap F = \emptyset$
    \end{enumerate}

    We start by proving the first property. Consider any fault scenario $F \in \mathcal{F}[2]$. Since $F \in \mathcal{F}[2]$, there is at least one process $p_i \in \mathcal{P}$ such that $\text{depth}(p_i)\geq 2$. Thus, there is at least one quorum $Q_i \in \mathcal{Q}_i$ such that all its members have depth at least one, that is, for each $p_{i'} \in Q_i$, $F \in \mathcal{F}_{i'}^*$.
    
     Consider now any two quorums $Q_1, Q_2 \in \mathcal{S}$. We know that $O(Q_1) \in \mathcal{Q}_j$ for $\mathcal{Q}_j \in \mathbb{Q}$ for some process $p_j$ and that $O(Q_2) \in \mathcal{Q}_k$ for $\mathcal{Q}_k \in \mathbb{Q}$ for some process $p_k$.  By the quorum consistency property of asymmetric quorum systems, we know that $Q_i \cap Q_j \neq \emptyset$. Let us denote with $p_m$ any process belonging to $Q_i \cap Q_j$. Similarly, we know that $Q_i \cap Q_k \neq \emptyset$. Let us denote with $p_n$ any process belonging to $Q_i \cap Q_k$. Let $A_m \in \mathcal{Q}_m$ be the quorum of $p_m$ contained within $Q_1$, that is, $A_m \subseteq Q_1$. Similarly, let $B_n \in \mathcal{Q}_n$ be the quorum of $p_n$ contained within $Q_2$. Since $p_m$ and $p_n$ have depth 1, from Corollary~\ref{lem:d1-intersect} it follows that $A_m \cap B_n \nsubseteq F$. Therefore, $Q_1 \cap Q_2 \nsubseteq F$. This proves the first property. 

     To prove the second property, consider again the party $p_i$ and its quorum $Q_i \in \mathcal{Q}_i$ where all processes have depth at least one. For each process $p_j$ in $Q_i$, it follows that $F \in \mathcal{F}_j^*$. This implies that there is a quorum $Q_j \in \mathcal{Q}_j$ such that $F \subseteq \mathcal{P} \setminus Q_j$. Let $S' \in \mathcal{S}$ be the symmetric quorum where each $p_j \in Q_i$ is replaced by the quorum $Q_j \in \mathcal{Q}_j$ such that $F \cap Q_j  = \emptyset$. Since $S'$ is the union of such quorums, and for each quorum its intersection with $F$ is empty, it follows that $S'\cap F=\emptyset$, which proves the second property. 
\end{proof}

\lthree*
\begin{proof}
    Since $p_i$ and $p_j$ both have depth at least 1, it holds $F \in \mathcal{F}_i^*$ and $F \in \mathcal{F}_j^*$. This implies $F \in \mathcal{F}_i^* \cap \mathcal{F}_j^*$. Then, the set $Q_i \setminus F$ intersects every quorum of $p_j$ by the quorum consistency property, and therefore contains a kernel for $p_j$.
\end{proof}

\lfour*
\begin{proof}
    Suppose there exists a compiler $\mathcal{C}$ that implements $\mathcal{C}[1]$, meaning that it is able to produce symmetric quorum systems $\mathcal{S}$ that satisfy the following properties

    \begin{enumerate}
        \item $\forall\ Q_1, Q_2 \in \mathcal{S}, \forall F \in \mathcal{F}[1]:  Q_1\cap Q_2 \nsubseteq F$
        \item $\forall\ F \in \mathcal{F}[1], \exists Q \in \mathcal{S}: Q\cap F = \emptyset$
    \end{enumerate}

    Consider the system with three processes $\mathcal{P} = \{p_1, p_2, p_3\}$ and asymmetric quorum system $\mathbb{Q} : $

    $\mathcal{Q}_1 = \bigl\{\{p_1, p_2\}\big\}$
    
    $\mathcal{Q}_2 = \bigl\{\{p_2, p_3\}\big\}$
    
    $\mathcal{Q}_3 = \bigl\{\{p_3, p_1\}\big\}$

    Note that this is a valid asymmetric quorum system since the canonical fail-prone system associated to it satisfies the $B^3$ property.

    We will show that there cannot exist a compiler $\mathcal{C}$ that produces a symmetric quorum system that satisfies the aforementioend properties. We proceed by contradiction. 
    
    Suppose that there is a compiler $\mathcal{C}$ that produces a symmetric quorum system that satisfies both properties. If the input is $\mathbb{Q}$, the compiler should produce the same output $\mathcal{S}$ regardless of the classification of processes as faulty or correct. 

    First, consider the case where $p_3$ is faulty. From the example quorum system it follows that $p_1$ has depth one ($p_1, p_2$ are correct) and $p_2$ has depth zero.
    
    Since $\mathcal{S} = \mathcal{C}(\mathbb{Q})$ should satisfy both properties, from the second one it is known that there must be a quorum $Q_1 \in \mathcal{S}$ such that $Q_1 \subseteq \{p_1, p_2\}$. 
    
    We first show that $|Q_1| > 1$. Suppose that $Q_1$ only has one element. Without loss of generality assume that it is $p_1$. By the first property it follows that all other quorums in $\mathcal{S}$ would need to contain $p_1$ (since $\emptyset \in \mathcal{F}[1]$). However, the failure scenario $F=\{p_1\}$ is also contained in $\mathcal{F}[1]$ ($p_2$ has depth one), so $\mathcal{S}$ would still need to satisfy both properties. But this would mean that if $p_1$ is faulty, $\forall Q\in \mathcal{S}, Q\cap F\neq \emptyset$, which breaks the second property. Therefore, $Q_1 = \{p_1, p_2\}$. 

    Using the same analysis for the cases when $p_1$ is faulty and when $p_2$ is faulty, it follows that $\mathcal{S}$ should also contain quorums $Q_2 = \{p_2, p_3\}$ and $Q_3 = \{p_1, p_3\}$. 

    Therefore, $\{Q_1, Q_2, Q_3\} \subseteq \mathcal{S}$. However, $\mathcal{S}$ cannot contain all these quorums, as this leads to breaking the properties required for the compiler. For example, $Q_1\cap Q_2 = \{p_1, p_2\} \cap \{p_2, p_3\} = \{p_2\}$, but $\{p_2\} \in \mathcal{F}[1]$, so the first property would be broken.

    Another way of reaching this conclusion is by noting that for the asymmetric quorum system $\mathbb{Q}$, $\mathcal{F}[1]$ does not satisfy the $Q^3$ property, which by Lemma~\ref{lem:canon} implies that no quorum system exists for it. 

    Therefore, since a compiler that implements $\mathcal{C}[1]$ must work for all asymmetric quorum systems and for all failure scenarios contemplated by $\mathcal{F}[1]$, we have shown that such a compiler cannot exist. Therefore, $\mathcal{C}[2]$ is the compiler with the lowest depth parameter that can be implemented. 
\end{proof}

\end{document}